\documentclass[review,a4paper]{elsarticle}
\usepackage{amsmath,amssymb,amsthm}
\usepackage{graphics,color}
\usepackage[multiple]{footmisc}

\newtheorem{theorem}{Theorem}
\newtheorem{corollary}[theorem]{Corollary}

\newtheorem{lemma}[theorem]{Lemma}

\newcommand{\N}{\mathbb{N}}

\newcommand{\prob}[2]{\mathbb{P}_{#1}\left(#2\right)}
\newcommand{\expect}[2]{\mathbb{E}_{#1}\left(#2\right)}
\newcommand{\Cycl}{R} 
\newcommand{\pcWh}{0.4030...} 
\newcommand{\lbTWh}{0.35} 
\newcommand{\ubTWh}{0.372} 

\begin{document}
\begin{frontmatter}
\title{Strict majority bootstrap percolation in the $r$-wheel\footnote{Supported by
CONICYT via Basal  (M.K.,P.M.,I.R), N\'ucleo Milenio ICM/FI P10-024F (M.K.,I.R.), Fondecyt 1130061 (I.R.), Fondecyt 1130709 (S.R.), ANR-09-BLAN-0164 (G.T.).}}

\author[dim]{M. Kiwi}
\author[dim]{P. Moisset de Espan\'{e}s}
\author[dim]{I. Rapaport\corref{cora}}
\cortext[cora]{Corresponding author (\texttt{rapaport@dim.uchile.cl}).}
\author[uai]{S. Rica}
\author[lama]{G. Theyssier} 
 \address[dim]{DIM, CMM (UMI 2807 CNRS), Universidad de Chile}
\address[uai]{Facultad de Ingenier\'{\i}a y Ciencias, Universidad Adolfo Ib\'a\~nez}
\address[lama]{LAMA, Universit\'e de Savoie, CNRS, France}

\begin{abstract}
In the strict Majority Bootstrap Percolation process each passive
vertex $v$ becomes active if at least  $\lceil \frac{deg(v)+1}{2}
\rceil$ of its neighbors are active (and thereafter never changes its
state). We address the problem of finding graphs for which a small
proportion of initial active vertices are likely to eventually make all
vertices active. We study the problem on a ring of $n$ vertices
augmented with a ``central'' vertex $u$. Each vertex in the ring, besides 
being connected to $u$,
is connected to its $r$ closest neighbors to the left and to the right. We
prove that if vertices are initially active with probability $p>1/4$
then, for large values of $r$, percolation occurs with probability
arbitrarily close to $1$ as $n\to\infty$. Also, if $p<1/4$,
then the probability of percolation is bounded away from $1$.
\end{abstract}

\end{frontmatter}

\section{Introduction}
Consider the following deterministic process on a graph $G=(V,E)$.
Initially, every vertex in $V$ can be either \emph{active} or
\emph{passive}. A passive vertex $v$ becomes active iff at least $k$ of
its neighbors are already  active; once active, a vertex never changes
its state. This process is known as
\emph{$k$-neighbor bootstrap percolation}~\cite{Ch79}. If at the end of
the process all vertices are active, then we say that the initial
set of active vertices \emph{percolates}.
We wish to determine the minimum ratio of initial active  
vertices needed to achieve percolation with high  
probability.  More precisely, suppose that the elements of the initial  
set of active vertices $A \subseteq V$ are chosen independently with  
probability $p$.  The problem is finding the least $p$ for
which percolation of $A$ is likely to occur.

Since its introduction by Chalupa et al~\cite{Ch79},  the bootstrap
percolation process has mainly been studied  in the $d$-dimensional   
grid $[n]^d=\{1,\ldots,n\}^d$~\cite{BBDM12}. The precise
definition of \emph{critical probability} that has been used is the
following: $$p_c([n]^d,k)= \inf\{p \in [0,1]: \prob{p}{\text{$A$
percolates $[n]^{d}$}} \geq 1/2\}.$$

In~\cite{BBDM12} it is proved that, for every $d \geq k \geq 2$,
$p_c([n]^d,k)= \left(\frac{\lambda(d,k)+o(1)}{\log_{(k-1)}n}\right)^{d-k+1},$
where $\lambda(d,k) < \infty$ are equal to the values of specific definite integrals for every $d \geq k \geq 2$. 
In the \emph{(simple) Majority Bootstrap Percolation (simple MBP)} process (introduced in~\cite{BBM})
each passive vertex $v$ becomes active iff 
  at least $\lceil \frac{deg(v)}{2} \rceil$ of its neighbors are active, 
  where $deg(v)$ denotes the degree of vertex $v$ in $G$.  
Note that for $[n]^d$, the critical probability 
  for simple MBP percolation corresponds to $p_c([n]^d,d)$, 
  which goes to $0$ as $n\rightarrow\infty$.

Here we introduce the \emph{strict Majority Bootstrap Percolation (strict MBP)}
process: each passive vertex $v$ becomes active iff at least
$\lceil \frac{deg(v)+1}{2} \rceil$ of its neighbors are active. Note
that if $deg(v)$ is odd, then strict and simple MBP
coincide. For $[n]^d$ the critical probability in strict
MBP $p_c([n]^d,d+1)$ goes to $1$. This holds because, in this case,
any unit hypercube starting with its $2^d$ corners passive will stay passive forever.

A natural problem is finding graphs for which the critical probability in the strict MBP
is small.  Results by Balogh and Pittel~\cite{BP}
imply that the critical probability  of the strict MBP
for random 7-regular graphs is 0.269. In~\cite{R11}, 
two families of graphs for which the critical
probability is also small (but higher than 0.269) are explored. The idea behind these
constructions is the following. Consider a regular graph of even degree 
$G$. Let $G*u$ denote the graph $G$ augmented with a single universal
vertex $u$. The strict MBP dynamics on $G*u$
has two phases. In the {\emph{first phase}}, assuming that vertex $u$ is
not initially active, the dynamics restricted to $G$ corresponds to the
strict MBP. If more than half of the vertices
of~$G$ become active, then the universal vertex $u$ also becomes active,
and the {\emph{second phase}} begins. In this new phase, the dynamics
restricted to $G$ follows the simple MBP (and
full activation  becomes much more likely to occur).

The two augmented  graphs studied in~\cite{R11} 
  were the \emph{wheel}  ${\textsc{WH}}_n = u * \Cycl_n$ 
  and the \emph{toroidal grid plus a universal vertex}  
  ${\textsc{TWH}}_n = u*\Cycl_n^2$
  (where  $\Cycl_n$ is the ring on $n$ vertices and
  $\Cycl_n^2$ is the toroidal grid on $n^2$ vertices).
For a family of graphs $\mathcal{G}=(G_{n})_n$, the following parameter
  was defined (as before, $A$ denotes the initial set of active vertices):
$$
p^+_c({\mathcal{G}})=\inf\left\{p\in [0,1] : 
  \liminf_{n\to\infty} \prob{p}{\text{$A$ percolates $G_n$ in strict MBP}}= 1  \right\}.
$$

Consider the families  $\mathcal{WH}=({\textsc{WH}}_n)_n$    and
$\mathcal{TWH}=({\textsc{TWH}}_n)_{n}$. It was proved in~\cite{R11} that
$p^+_c({\mathcal{WH}})= \pcWh$.   For the toroidal case it was shown
that $\lbTWh \leq p^+_c({\mathcal{TWH}})\leq \ubTWh.$ Computing the critical probability for the wheel is 
trivial. Nevertheless, if we increase the \emph{radius} of the
vertices, then the situation becomes much
more complicated. More precisely, let $\Cycl_n(r)$ be the ring where
every vertex is connected to its $r$ closest vertices to the left and 
to its $r$ closest vertices to the right.  Here we study the strict MBP process
in a  generalization of the wheel that we call \emph{$r$-wheel} 
${\textsc{WH}}_n(r) = u * \Cycl_n(r)$. Our main result is the following:
\begin{theorem}
The limit of $p^+_c({\mathcal{WH}}(r))$, as $r\to\infty$,
  exists and equals $1/4$.
\end{theorem}

\section{Preliminary results}
We start by showing that we can reduce our problem to the issue
of whether a single fixed (non-universal) vertex eventually
 becomes active.
  
\begin{lemma}\label{lem:ringenough}
Let $0<p<1$ be the probability for a vertex to be initially active.  Let
$r$ be a positive integer. Denote by $p_W(n,r,p)$ the percolation
probability of    the \emph{$r$-wheel} and denote by $p_R(n,r,p)$ the
probability that the strict majority on $\Cycl_n(r)$ ends up with (strictly)
more active than passive vertices. Then,
\begin{align*}
    \liminf_{n\rightarrow\infty}p_R(n,r,p) &\ \leq\ \liminf_{n\rightarrow\infty}p_W(n,r,p), \\
    \limsup_{n\rightarrow\infty}p_W(n,r,p)&\ \leq\ p + (1-p)\cdot\limsup_{n\rightarrow\infty}p_R(n,r,p).
\end{align*}
\end{lemma}
\begin{proof}
Note that for $\epsilon>0$ we can choose $n$ large enough so that  
the probability that at least one block of $r$ consecutive vertices are  
initially active is larger than ${1-\epsilon}$, in which case  
percolation occurs iff the universal vertex becomes active  
during the evolution.  We deduce the first
inequality by taking $\epsilon$ arbitrarily small. 
Note now that the universal vertex is active
when the dynamics stabilizes only if  it was either  already active
initially (probability $p$) or if it was initially passive and the
dynamics on the ring $\Cycl_n(r)$ produces more than $n/2$ active vertices.
\end{proof}

The vertices of the ring $\Cycl_{n}$ will be denoted as 
$0,1,\ldots,n-1$, starting at some arbitrary vertex (arithmetic over
vertex indices will always be modulo $n$). The positive integer
  $r$ will be called the \emph{radius}.

Lemma~\ref{lem:ringenough} shows that we can study the ring $\Cycl_n(r)$ and
its dynamics to derive results about the $r$-wheel. Now, fix some arbitrary
vertex and consider the $0$-$1$ random variable $X_i(n,r)$ giving the state of
vertex $i$ after stabilization of the dynamics ($X_{i}(n,r)=0$ if the 
  state is passive, and $X_{i}(n,r)=1$ if it is active).
Next, we show how to bound $p_R(n,r,p)$ in terms of
  $\expect{p}{X_0(n,r)}$.

\begin{lemma}\label{lem:thanksmarkov}
Let $0<p<1$, $n \in \N^+$, and $r$ a fixed radius.  Then,

$2\expect{p}{X_0(n,r)}-1\leq p_R(n,r,p)\leq 2\expect{p}{X_0(n,r)}.$
\end{lemma}
\begin{proof}
  By definition ${p_R(n,r,p)=\prob{p}{\sum_i X_i(n,r)> n/2}}$. By Markov's inequality we then have
  $\prob{p}{\sum_i X_i(n,r)> n/2} \leq \frac{2}{n}\expect{p}{\sum_i X_i(n,r)}.$
  Using linearity of expectation and the fact that all $X_i(n,r)$ are equally distributed (symmetry of the ring), we deduce ${p_R(n,r,p)\leq 2\expect{p}{X_0(n,r)}}$. The lower bound is
  obtained in the same way considering again Markov's inequality but
  for the (again positive) random variable ${n-\sum_i X_i(n,r)}$. More precisely:
  \[p_R(n,r,p) = 1 - \prob{p}{n-\sum_iX_i(n,r)>n/2}\geq 1 - \frac{2}{n}\expect{p}{n-\sum_iX_i(n,r)}\]
\end{proof}

\section{Lower bound on $p^{+}_{c}(\mathcal{WH}(r))$}
We will assume $n>2r+1$ and that
  the initial state of the universal vertex~$u$ is passive.
Let $0<p< 1/2$ and $q=1-p$. The starting configuration 
  $\sigma=(\sigma_0,\ldots,\sigma_{n-1})$, where vertex $j$ is initially
  active (respectively passive) if and only if $\sigma_j= 1$ (respectively
  $\sigma_j=0$), 
  occurs with probability $p^{\sum_j{\sigma_j}} q^{n-\sum_j{\sigma_j}}$. 
We write $X_0$ instead of $X_0(n, r)$. 
Conditioning on $\sigma_0$,
\begin{align}
\prob{p}{X_0=1} 
  \leq p + \prob{p}{X_0=1 | \sigma_0=0}. \label{eqn:first}
\end{align}

We say there is a \emph{wall} located 
  $\ell>0$ vertices to the left of vertex $0$
  if $\sigma_{-\ell}=1$, $\sigma_{-\ell-1}=\sigma_{-\ell-2}=\ldots
  =\sigma_{-\ell-(r+1)}=0$.
Similarly, we say there is a wall located at 
  $\ell>0$ vertices to the right of vertex $0$
  if $\sigma_{\ell}=1$, $\sigma_{\ell+1}=\sigma_{\ell+2}=\ldots
  =\sigma_{\ell+(r+1)}=0$.
Let $L$ (respectively $R$) 
  be the smallest positive $\ell$ such that there is a 
  wall located $\ell$ vertices to the left (respectively right) of vertex $0$ 
  (if a wall does not occur, let $L=R=n$). 
For $0<\Delta<n$ to be fixed later, and since $L$ and $R$ are 
  identically distributed, we have that:
\begin{align}
& \prob{p}{X_0=1 | \sigma_0=0} \nonumber \\ 
& \quad \leq 2\cdot\prob{p}{X
_0=1, R\geq \Delta | \sigma_0=0} + \prob{p}{X_0=1 \wedge L, R<\Delta | \sigma_0=0} \label{eqn:second}
\end{align}

Summarizing, to bound $\expect{p}{X_0}=\prob{p}{X_0=1}$ we can
bound the two terms in the right hand side of (\ref{eqn:second}). 
The proof of next lemma is straightforward.

\begin{lemma}\label{lem:expectation}
For $0<p<1$ and positive integers $a,r$,

$\expect{p}{R | \sigma_0=\sigma_1=\ldots=\sigma_{a-1}=0, \sigma_{a}=1} 
  \leq q^{-r}(aq^{r}+1/(pq)).$
\end{lemma}
\begin{proof}
Consider a Markov chain with states labeled $0,1,\ldots,r+1$ where 
  for all $s\leq r$, the probability of going from state $s$ to
  $0$ (respectively $s$ to $s+1$) is $p$ (respectively $q$), and once
  state $r+1$ is reached, the Markov chain stays there forever.
For $s\in\{0,\ldots,r+1\}$, let $N_s$ be the number of 
  steps it takes the Markov chain to reach state $r+1$ when
  it starts at state $s$.
Note that 
\[
\expect{}{R| \sigma_0=\sigma_1=\ldots=\sigma_{a-1}=0, \sigma_{a}=1} 
  \leq a+\expect{}{N_{0}}.
\]
Moreover, $\expect{}{N_{r+1}}=0$, and 
  $\expect{}{N_s} = 1+q\cdot\expect{}{N_{s+1}}+p\cdot\expect{}{N_0}$
  for all $0<s\leq r$.
Thus, for all $0\leq s\leq r+1$,
\[
\expect{}{N_{0}} 
  = \sum_{j=1}^{s} \frac{1}{q^j} + \expect{}{N_s} 
  = \sum_{j=1}^{r+1} \frac{1}{q^j}
  \leq \frac{1}{pq^{r+1}}.
\]
Putting everything together yields the  result.
\end{proof}

\begin{corollary}\label{cor:expectation}
For $0<p<1$ and  positive integers $r,\Delta$,

$\prob{p}{X_0=1, R\geq\Delta | \sigma_0=0} 
  \leq \frac{1}{\Delta}\cdot q^{-r}(rq^{r}+1/(pq)).$
\end{corollary}
\begin{proof}
If vertex $0$ eventually becomes active, it
must be the case that initially it did not belong to a block of $r+1$
consecutive passive vertices. Thus, if $X_0=1$ and $\sigma_0=0$, then there
must exist  a positive integer $a$ such that $a\leq r$,
$\sigma_1=\sigma_2=\ldots=\sigma_{a-1}=0$, and $\sigma_a=1$. For
brevity, we will denote this particular array of outcomes for the
$\sigma$'s as $C_a$. By Markov's inequality,
\begin{align*}
 \prob{p}{X_0=1, R\geq\Delta | \sigma_0=0} 
  & \leq \sum_{a=1}^{r} 
       \prob{p}{R\geq\Delta | C_a}
       \prob{p}{C_a} \leq \frac{1}{\Delta}\max_{a=1,\ldots,r} \expect{p}{R | C_a}.
\end{align*}
The desired conclusion follows from Lemma~\ref{lem:expectation}.
\end{proof}

\begin{lemma}\label{lem:dead-block} 
For $0<p<1/2$ and positive integers $r,\Delta$,

$\prob{p}{X_0=1 \wedge L,R<\Delta | \sigma_0=0} \leq 2\Delta\left(4pq\right)^{r}.$
\end{lemma}
\begin{proof}
Suppose the closest wall to the left (respectively right) of vertex
$0$ is at vertex $-a$ (respectively $b$). Furthermore, suppose vertex $0$ is
passive. Note that for vertex $0$ to eventually become active, it must be
the case that some passive vertex $i$ for $-a < i < b$ must necessarily
become active the first time the strict majority dynamics is applied.
Hence, if $S_i$ denotes the number of $j$'s, $j\neq i$ and    $i-r\leq
j\leq i+r$, for which  vertex $j$ initially takes the value $1$, then
\begin{align*}
\prob{p}{X_0=1 \wedge L,R<\Delta | \sigma_0=0} & \leq \prob{p}{\exists i, -\Delta<i<\Delta \text{s.t.} S_i\geq r+1} \\
  & \leq 2\Delta\max_{i:-\Delta<i<\Delta}\prob{p}{S_i\geq r+1}\,.
\end{align*}

However, a Chernoff bound
 tells us that, for $t=1/2-p \leq (r+1)/(2r)-p$,
\begin{align*}
\prob{p}{S_i\geq r+1}  & \leq \prob{p}{S_i\geq (p+t)\cdot 2r}  \leq \left(\left(\frac{p}{p+t}\right)^{p+t} \left(\frac{q}{q-t}\right)^{q-t}\right)^{2r} \leq \left(4pq\right)^{r}
\end{align*}
Putting everything together yields the conclusion.
\end{proof}

\begin{theorem}\label{thm:upperbound}
  For all $0<p<1/4$ there exists a large enough
  integer $r_0=r_0(p)$ such that 
  if $r\geq r_0$ and $n>2r+1$, then $\expect{p}{X_0(n,r)}<1/4$.
\end{theorem}

\begin{proof}
  Let $r'_0=r'_0(p)$ be such that $r\geq r'_0$ implies that $rq^r\leq
  1/pq$, and let $C=8r/(pq)$.  By Corollary~\ref{cor:expectation},
  $ \prob{p}{X_0=1, R\geq Cq^{-r} | \sigma_0=0} \leq
  \frac{2}{Cpq}=\frac{1}{4r}.$ By~\eqref{eqn:first}, and fixing $\Delta=Cq^{-r}$
  in~\eqref{eqn:second}, and Lemma~\ref{lem:dead-block}, we obtain
  that:
  
 $\prob{p}{X_0=1} \leq p + \frac{1}{2r}+ 2Cq^{-r}(4pq)^{r} = p +
  \frac{1}{2r} + 2C(4p)^{r}.$
  
\noindent%
Hence, for $p<1/4$ there exists a large enough positive integer
  $r_0=r_0(p)\geq r'_{0}(p)$ so that if $r\geq r_0$, then $p +
  \frac{1}{2r} + 2C(4p)^{r} < p+\frac{1}{r} < 1/4$.
\end{proof}
Theorem~\ref{thm:upperbound},
  Lemma~\ref{lem:thanksmarkov}, and Lemma~\ref{lem:ringenough} 
  yield the following:
\begin{corollary}\label{cor:liminf}
$\liminf_{r\to\infty} p^+_c({\mathcal{WH}}(r)) \geq 1/4$.
\end{corollary}

\section{Upper bound on $p^{+}_{c}(\mathcal{WH}(r))$}
  Consider a simplified process with three states on the
  one-dimensional integer lattice $\mathbb{Z}$: 
  (i) $w$, a wall, (ii) $s$, a spreading state, and 
  (iii) $e$, an empty lattice point.
Let sites in state $w$ and $s$ remain so
  forever, and in consecutive rounds let sites in state $e$ with at 
  least one neighbor in state $s$ update to state $s$. 
Let $p_w$, $p_s$ and $p_e$ be positive initial probabilities of 
  states $w$, $s$ and $e$ respectively, where $p_w + p_s + p_e = 1$.
Each lattice point is initially assigned a state, independent of the 
  other lattice point states. 
\begin{lemma}\label{lem:simple}
The probability that lattice point $0$ is eventually in state 
  $s$ is greater than $1/(1+p_w/p_s)$.
\end{lemma}
\begin{proof}
Define $s_L$ (respectively
$s_R$) to be state of the closest lattice point on the left
(respectively right) of $0$ whose state is not $e$. 
Since $p_e=1-p_w-p_s<1$, 
  both $s_L$ and $s_R$ are well-defined with probability $1$.
Let $E$ be the event that lattice 
  point $0$ is eventually in state $s$, and denote by $P$ its probability
  of occurring.
Note that for $E$ to occur, either the lattice point 
  $0$ is initially in state $s$, or it is initially in state~$e$ and at least 
  one of the lattice points $s_L$ or $s_R$ is initially in
  state $s$.
Hence, recalling that $p_e=1-p_w-p_s<1$, 
\begin{align*}
P & = p_s + p_{e}\sum_{i\geq 0,j\geq 0} p_e^{i+j}(p_s^2+2p_wp_s) 
  = p_s+\frac{p_s(p_s+2p_w)p_e}{(1-p_e)^2} \\
  & \geq p_s\left(1+\frac{1-p_w-p_s}{p_w+p_s}\right) = \frac{1}{1+p_w/p_s}.
\end{align*}
\end{proof}

We now consider again the strict MBP process in the ring
$\Cycl_n(r)$ and reduce it to the aforementioned three-state  model
as follows: Fix some length $\ell$ and partition the vertices
of $\Cycl_{n}(r)$ into length $\ell$ blocks (i.e., sets of $\ell$
consecutive vertices, where $n=t\ell$). Let
$W_{\ell,r}$ be the set of all possible blocks of length $\ell$ that 
contain $r+1$ consecutive passive vertices.  
Also, let $S_{\ell,r}$ be the set of all blocks of length $\ell$ that
do not contain $r+1$ consecutive passive vertices and that, for any
state configuration for vertices not contained in the block, all
the vertices belonging to the block eventually become active when
applying the strict majority dynamics.
Any block in $W_{\ell,r}$ is  a wall in $\Cycl_{n}(r)$ and any block in   
$S_{\ell,r}$ is  a spreading state.  Any other  block is an empty state.  Let $\mu(W_{\ell,r})$ 
(respectively~$\mu(S_{\ell,r})$)
be the  probabilityr
that an arbitrary block belongs to $W_{\ell,r}$  (respectively~$S_{\ell,r}$).
The following lemma is not difficult to prove:
\begin{lemma}\label{lem:linksimple}
For $0<p<1$ and positive integers $r,\ell$,
 
(i) $\liminf_{n\rightarrow\infty} \expect{p}{X_0(n,r)} \geq 1/\bigl(1+\mu(W_{\ell,r})/\mu(S_{\ell,r})\bigr)$.
 
(ii) If $\ell\geq r+1$, then $\mu(W_{\ell,r})\leq \ell q^{r+1}$.
\end{lemma}

We will now find a lower bound for $\mu(S_{\ell,r})$.
The goal is to prove that $\frac{\mu(W_{\ell,r})}{\mu(S_{\ell,r})}$ goes
to 0 when $r\to\infty$. For that purpose we denote, for any $0$-$1$ word $v$, by
$|v|_0$ (respectively $|v|_1$)    the number of occurrences of symbol
$0$ (respectively~$1$) in $v$,   and denote the $i$-th character of $v$
by $v_i$. We set $\ell=2r+1$ and consider the set $T_r$
of binary words $v$ of length $\ell$ satisfying the following
properties:   (1) $|v|_1=r+1$,$|v|_0=r$;  (2) $v_0=v_{2r}=1$,
$v_{r}=0$ and   (3) the word $w=w_1\cdots w_{r-1}$ of length $r-1$  
over alphabet    $\{0,1\}^2$ defined by $w_i=(v_i,v_{i+r})$   is a
generalized Dyck word~\cite{Duchon00} associated to the weight function
$\omega(a,b) = +1$ if $(a,b)=(0,0)$, $\omega(a,b) = -1$  if
$(a,b)=(1,1)$ and
$\omega(a,b) = 0$ otherwise. 
I.e., $\omega$ satisfies the following two conditions:
(i)  $\sum_{i=1}^j\omega(w_i) \geq 0$ for all $1\leq j\leq r-1$,
and
(ii) $\sum_{i=1}^{r-1}\omega(w_i) = 0$.

\begin{lemma}\label{lem:include-generalized}
If $r$ is a positive integer and $\ell=2r+1$, then 
  $\left|T_r\right|\leq\left|S_{\ell,r}\right|$.
\end{lemma}
\begin{proof}
Consider some $v\in T_r$ and denote by $w$ the word of length $r-1$ over 
  the alphabet $\{0,1\}^2$ as defined above. 
We first consider successively each vertex $i$ of $\Cycl_n(r)$, for $i=r$ to 
  $i=2r-1$, and apply the strict majority dynamics to it. 
Initially, the state of vertex $i$ is $v_i$.
During this first 
  sequence of updates, we denote by $n_i$ the number of $1$s in 
  the neighborhood of vertex $i$ at the time this vertex is considered (i.e., 
  we take into account updates of vertices $j<i$ which possibly occurred 
  before in the sequence). 
Since $v\in T_{r}$, we have that $|v|_1=r+1$, so
  at the beginning of the process $i=r$ and $n_r=r+1$.
We claim that for all $i$ with $r\leq i<2r-1$, when we consider 
  vertex $i+1$ in the process we have $n_{i+1}-n_i = \omega(w_{i+1-r})$ and 
  the state of vertex $i$ is $1$.
  This claim is deduced by induction from the fact that 
  $i-r$ is the only vertex of index less than $i+1$ in the  
  symmetric difference of the neighborhoods of vertices $i$ and $i+1$, hence 
  $n_{i+1}-n_i = \delta_+ - \delta_-$
  where $\delta_+=1$ if vertex $i$ was updated from $0$ to $1$ in the previous step and $0$ otherwise, and $\delta_-=1$ if $v_{i-r}=1$ 
  and $0$ otherwise. Using the induction hypothesis for all $j\leq i$,
  by the Dyck property of $w$, we have  $n_{i} = n_r + \sum_{j=1}^{i-r} \omega(w_j) \geq n_r=r+1.$
Hence, if vertex $i$ was $0$ before being considered ($v_i=0$), it updates
to $1$ when considered. In any case vertex $i$ is $1$ once considered in
the process. Moreover, we have $\delta_+=1$ iff $v_i=0$ and
we deduce that $\delta_+ - \delta_-=\omega(w_{i+1-r})$, thus
establishing the claim.

Now, from the claim, we deduce that all vertices from $i=r$ to $i=2r$ are
  active when the sequence of updates ends.
Then, we consider a new sequence of updates from vertex $i=r-1$ to 
  vertex $i=1$ (successively). 
Trivially,
  when vertex $i$ is considered, its neighborhood
  contains at least $r+1$ active vertices (because the neighborhood 
  of vertex $i$ contains the active vertices
  $0$ and vertices $i+1,\ldots,i+r$). 
Hence, all vertices will be active 
  at the end of this second sequence of updates.
This completes the proof of the lemma.
\end{proof}

\begin{theorem}\label{thm:lowerbound}
If $1>p>1/4$, then there is an $r_0=r_0(p)$ such that for all $r\geq r_0$ we have: $\lim_{n\rightarrow\infty} \expect{p}{X_0(n,r)} = 1.$
\end{theorem}
\begin{proof}
First, we claim that there is some positive rational function $\phi(\cdot)$ such that 
  $|T_r| \geq \phi(r)\cdot 4^r$ for all $r$. To prove this,
  choose $r = 4k+1$ (this is without loss
  of generality since ${r\mapsto|T_r|}$ is increasing). It is
  straightforward to associate injectively a word $v\in T_r$ to any word
  $w$ of length $r-1 = 4k$ over alphabet $\{0,1\}^2$ which is a
  generalized Dyck word associated to the  weight function $\omega$ as
  defined before. To obtain a lower bound on the
  number of such generalized Dyck words, we consider the subset $U_k$ of
  words $w$ of length $4k$ over the  alphabet $\{0,1\}^2$ and such that 
  $\left|\bigl\{i : \omega(w_i) = +1\bigr\}\right| = 
   \left|\bigl\{i : \omega(w_i) = -1\bigr\}\right|= k$ and
  $\left|\bigl\{i : \omega(w_i) = 0\bigr\}\right|= 2k$. The set $U_k$
can be generated, up to a straightforward encoding, by considering
classical Dyck words of length $2k$ (weights $+1$/$-1$) interleaved by
binary words of size~$2k$. Therefore, $|T_r| \geq D_{2k}\cdot {4k\choose
2k}2^{2k}.$
Using classical results about Catalan numbers and Stirling's formula, 
  for some positive rational functions $\phi_1$ and $\phi_2$ we have 
  $|D_{2k}| \sim \phi_1(k)\cdot 4^k$ and 
  ${4k \choose 2k}\sim \phi_2(k)\cdot 4^{2k}$. 
It follows that there is some positive rational function $\phi$ such that 
  $|T_r| \geq \phi(r)\cdot 4^r$. This establishes the claim.

By Lemma~\ref{lem:include-generalized}, for some positive rational function $\phi$ it holds that
  $\mu(S_{\ell,r})\geq \phi(r)\cdot q^rp^{r+1}4^r$.
By Lemma~\ref{lem:linksimple} $(ii)$, 
  $\frac{\mu(W_{\ell,r})}{\mu(S_{\ell,r})}$ is asymptotically less than
  $\Phi(r)\cdot(4p)^{-r}$, where $\Phi(\cdot)$ is another positive rational
  function. If  $p> 1/4$ then
  $\mu(W_{\ell,r})/\mu(S_{\ell,r})$ goes to $0$ when $r\to\infty$.
The theorem follows from  Lemma~\ref{lem:linksimple} $(i)$.
\end{proof}

It follows that  $\limsup_{r\to\infty} p^+_c({\mathcal{WH}}(r)) \leq
1/4$ and, invoking Corollary~\ref{cor:liminf}, we conclude that $\lim_{r\to\infty}
p^+_c({\mathcal{WH}}(r)) = 1/4$.

%


\begin{thebibliography}{7}
\bibitem{BBDM12} 
J.~Balogh, B.~Bollob\'{a}s, H.~Duminil-Copin and R.~Morris.
The sharp threshold for bootstrap percolation in all dimensions.
 Trans.~Amer.~Math.~Soc.~364:2667-2701 (2012).

\bibitem{BBM} J.~Balogh, B.~Bollob\'{a}s and R.~Morris. 
Majority  bootstrap percolation on the hypercube. 
Combin.~Probab.~Comput.~18:17-51  (2009).


\bibitem{BP} 
J.~Balogh and B.~Pittel. 
Bootstrap percolation on the
random regular graph. Random Structures \& Algorithms 30(1-2):257-286 (2007).

\bibitem{Ch79} J.~Chalupa, P.~L.~Leath and G.~R.~Reich.  Bootstrap percolation on a Bethe lattice. J.~Phys.~C: Solid State Phys.~12:L31-L35 (1979).

\bibitem{Duchon00}
P.~Duchon. 
On the enumeration and generation of generalized Dyck words.
Discrete Math.~225(1-3):121-135 (2000).

\bibitem{R11} I.~Rapaport,  K.~Suchan, I.~Todinca and J.~Verstraete.
On dissemination thresholds in regular and irregular graph classes.
 Algorithmica 59:16-34 (2011).
\end{thebibliography}
\end{document}